\documentclass[authoryear]{article}
\usepackage{authblk}
\usepackage{natbib}

\usepackage[T1]{fontenc}
\usepackage[utf8]{inputenc}
\usepackage{pdfpages}
\usepackage{amsmath,amssymb}

\usepackage{amsthm}

\newtheorem{theorem}{Theorem}

\newtheorem{lemma}[theorem]{Lemma}

\newtheorem{corollary}[theorem]{Corollary}
\newtheorem{observation}[theorem]{Observation}

\theoremstyle{definition}

\usepackage{mathtools}
\usepackage{tikz}
\usepackage{graphicx}

\usetikzlibrary{shapes}
\usetikzlibrary{backgrounds}
\usetikzlibrary{decorations.pathmorphing}
\usetikzlibrary{decorations.pathreplacing}
\usetikzlibrary{calc}
\usetikzlibrary{fit}
\usetikzlibrary{positioning}

\usepackage{xspace}
\usepackage{enumerate}
\usepackage{mathdots} 
\usepackage{times}
\usepackage[shortlabels]{enumitem}

\newdimen\arrowsize
\newlength{\arrowlength}
\newlength{\arrowangle}
\newlength{\arrowthickness}

\pgfarrowsdeclare{slim}{slim}
{
\arrowsize=0.7pt
\advance\arrowsize by .5\pgflinewidth
\pgfarrowsleftextend{-4\arrowsize-.5\pgflinewidth}
\pgfarrowsrightextend{.5\pgflinewidth}
}
{
\advance\arrowsize by .5\pgflinewidth
\pgfsetdash{}{0pt} 
\pgfsetroundjoin 
\pgfsetroundcap 
\pgfpathmoveto{\pgfpointpolar{\arrowangle}{\arrowlength}}
\pgfpathlineto{\pgfpointorigin}
\pgfusepathqstroke
\pgfpathmoveto{\pgfpointpolar{\arrowangle}{\arrowlength}}
\pgfpathcurveto{\pgfpointpolar{\arrowangle*1.04}{\arrowthickness+0.5*(\arrowlength-\arrowthickness)}}%
{\pgfpointpolar{\arrowangle*1.1}{\arrowthickness+0.3*(\arrowlength-\arrowthickness)}}
{\pgfpointpolar{180}{\arrowthickness}}
\pgfpathlineto{\pgfpointorigin}
\pgfusepathqfill
\pgfpathmoveto{\pgfpointpolar{-\arrowangle}{\arrowlength}}
\pgfpathlineto{\pgfpointorigin}
\pgfusepathqstroke
\pgfpathmoveto{\pgfpointpolar{-\arrowangle}{\arrowlength}}
\pgfpathcurveto{\pgfpointpolar{-\arrowangle*1.04}{\arrowthickness+0.5*(\arrowlength-\arrowthickness)}}%
{\pgfpointpolar{-\arrowangle*1.1}{\arrowthickness+0.3*(\arrowlength-\arrowthickness)}}
{\pgfpointpolar{180}{\arrowthickness}}
\pgfpathlineto{\pgfpointorigin}
\pgfusepathqfill
}

\tikzstyle{vertex}=[circle,inner sep=2.5,minimum size =2mm,semithick,fill=black!20, draw=black]
\tikzstyle{holevertexB}=[circle,inner sep=-0.7,minimum size =2mm,semithick,fill=white, draw=black]
\tikzstyle{holevertexA}=[circle,inner sep=2.5,minimum size =2mm,semithick,fill=white, draw=black]
\tikzstyle{smallcircle}=[circle,inner sep=1.5,fill=white, draw=black]
\tikzstyle{point}=[circle,inner sep=1,fill=black, draw=black]
\tikzstyle{path}=[-slim,thin,rounded corners] 
\tikzstyle{path1}=[-slim,thin,decorate,%
decoration={snake,amplitude=6pt,segment length=#1,%
pre length=#1/20,post length=#1/20}]
\tikzstyle{path2}=[-stealth,thin,decorate,%
decoration={snake,amplitude=6pt,segment length=#1,%
pre length=#1/10,post length=#1/10}]
\tikzstyle{brace}=[thin,decorate,decoration=brace]
\tikzstyle{ie}=[thin,dashed,gray]

\colorlet{fillA}{gray!50}
\colorlet{fillB}{gray!15}

\input{abbrs}

\DeclareMathOperator{\Reach}{Reach}

\DeclareMathOperator{\var}{var}
\DeclareMathOperator{\lev}{lev}

\renewcommand{\mid}{\ensuremath{:}}

\renewcommand{\phi}{\varphi}
\renewcommand{\epsilon}{\varepsilon}
\renewcommand{\theta}{\vartheta}
\newcommand{\0}{\emptyset}

\newcommand{\perpcdot}{{\perp}{\kern0.1em{\cdot}\kern0.1em}}

 \newcommand{\Pos}{\mathrm{Pos}}
 \newcommand{\Moves}{\mathrm{Moves}}

\newcommand{\Mon}{\textup{Mon}}

\DeclareMathAlphabet{\mathsc}{OT1}{cmr}{m}{sc}

\newcommand{\ptime}{\ensuremath{\mathsc{Ptime}}\xspace}
\newcommand{\np}{\ensuremath{\mathsc{NP}}\xspace}
\newcommand{\conp}{\ensuremath{\text{co-}\mathsc{NP}}\xspace}
\newcommand{\pspace}{\ensuremath{\mathsc{PSpace}}\xspace}
\newcommand{\aptime}{\ensuremath{\mathsc{APTime}}\xspace}
\newcommand{\nptime}{\ensuremath{\mathsc{NPtime}}\xspace}
\newcommand{\DAGWprob}{\ensuremath{\mathsc{DAGW}}\xspace}
\newcommand{\taut}{\ensuremath{\mathsc{Tautology}}\xspace}

\newcommand{\QBF}{\ensuremath{\mathrm{QBF}}\xspace}

\newcommand{\cupdot}{\mathbin{\mathaccent\cdot\cup}}

\DeclareMathOperator{\MCGAME}{MCgame}

\newcommand{\tw}{tree {}width\xspace}

\newcommand{\Tw}{Tree {}width\xspace}
\newcommand{\dtw}{directed tree {}width\xspace}

\newcommand{\dagw}{DAG-{}width\xspace}
\newcommand{\Dagw}{DAG-{}width\xspace}
\newcommand{\kw}{Kelly-{}width\xspace}

\newcommand{\dw}{D-{}width\xspace}

\DeclareMathOperator{\DAGW}{DAG\text{-}w}

\newcommand{\ie}{i.e.\@\xspace}

\newcommand{\eg}{e.g.,\xspace}

\usepackage{hyperref}
\hypersetup{colorlinks=true,linkcolor=blue!60!black,citecolor=blue!70!black}

\begin{document}
  \title{DAG-width is PSPACE-complete\thanks{The research of all authors has
    been supported by the European Research Council (ERC) under the European Union’s Horizon 2020 research and
    innovation programme (ERC consolidator grant DISTRUCT, agreement No 648527).}}
  \author[ ]{Saeed Akhoondian Amiri}
  \author[ ]{Stephan Kreutzer}
  \author[ ]{Roman~Rabinovich}
  \affil[ ]{Logic and Semantics, Technische Universit\"at Berlin}
  \affil[ ]{\texttt{\{saeed.amiri,stephan.kreutzer,roman.rabinovich\}\\@tu-berlin.de}}
  \date{}

\maketitle
  
\begin{abstract}

  Berwanger et al.\@ show in~\cite{BerwangerDawHunKreObd12} that for
  every graph~$G$ of size~$n$ and \dagw~$k$ there is a DAG
  decomposition of width~$k$ and size $n^{O(k)}$. This gives a
  polynomial time algorithm for determining the \dagw of a graph for
  any fixed~$k$. However, if the \dagw of the graphs from a class is
  not bounded, such algorithms become exponential. This raises the
  question whether we can always find a DAG decomposition of size
  polynomial in~$n$ as it is the case for \tw and all generalisations
  of \tw similar to \dagw.

  We show that  there is an infinite class of graphs such that every DAG
  decomposition of optimal width has size super-polynomial in~$n$ and,
  moreover, there is no polynomial size DAG decomposition 
  which would approximate an optimal decomposition up to
  an additive constant.

  In the second part we use our construction to prove that deciding
  whether the DAG-width of a given graph is at most a given constant
  is \pspace-complete.

\end{abstract}


\section{Introduction}

In the study of hard algorithmic problems on graphs,
methods derived from structural graph theory have proved to be a
valuable tool. The rich theory of special classes of graphs developed
in this area has been used to identify classes of graphs, such as
classes of bounded tree width or clique width,
on which many computationally hard problems can be solved
efficiently. 
Most of these classes are defined by some structural
property, such as having a tree decomposition of low width, and this structural 
information can be exploited
algorithmically. 

Structural properties of classes of graphs such as tree width, clique
width, definability by excluded minors etc.\@ studied in this context relate to
undirected graphs.
However, in various applications in computer science, directed graphs are
a more appropriate model. 
Given the enormous success
width parameters 
had for problems defined on undirected graphs, it is natural to ask
whether they can also be used to analyse the complexity of hard
algorithmic problems
on digraphs.
While in principle it is 
possible to apply the structure theory for undirected graphs to directed
graphs by ignoring the direction of edges, this implies a significant
information loss. Hence,
for computational problems whose instances are directed graphs,  
methods based on the structure theory for undirected graphs may be
less useful.

\Tw is one of the most successful structural complexity measures. It has several
characterisations coming from seemingly unrelated notions, \eg by
elimination orders or cops and robber games. \Tw is also deeply
connected to graph minors and has numerous algorithmic applications.
The result of several approaches to generalise \tw to digraphs was a
number of structural complexity measures for digraphs. 
Reed \cite{Reed99} and Johnson, Robertson, Seymour and
Thomas~\cite{JohnsonRobSeyTho01} introduced the
concept of \emph{directed tree width} and showed that the $k$-disjoint
paths problem and more general linkage problems can be solved in
polynomial-time on classes of digraphs of bounded directed
tree width. Following this initial proposal, several alternative
notions of width measures for sparse classes of digraphs have been
presented, for instance \emph{directed path width} (see
\cite{Barat06}, initially proposed by Robertson, Seymour and Thomas),
\emph{D-width}~\cite{Safari05},
\emph{DAG-width}~\cite{BerwangerDawHunKreObd12} and
\emph{Kelly-width}~\cite{HunterKre08}. 

In this work we concentrate on \dagw. It distinguishes itself in its
particularly simple definition. DAG decompositions have a clear
structure and the definition of the cops and robber games
characterising \dagw is a straight forward and natural generalisation
of the corresponding game for \tw. However, we show some disadvantages
of \dagw.

A crucial task in designing efficient algorithms for some problems on graphs where some
width is bounded is to find a decomposition of the given graph of
small width. Such decompositions, usually trees or DAGs (directed acyclic graphs), are used to
solve some problem recursively following the decomposition. For \tw and
\dtw one can decompose the graph in fixed parameter tractable
time. For \dw and for \kw such algorithms are not known, but there is
always a decomposition of polynomial size, so it can be found at least
 non-deterministically. 
Only the complexity of \dagw was left as an open problem
as it was not known whether every digraph has a decomposition of
polynomial size.

We show here that deciding the \dagw of a
digraph is not only not in \np (under standard complexity theoretical
assumptions), it is in fact \pspace-complete. In terms of the \dagw
game this exhibits the worst case complexity of such games. This
result is quite unexpected and especially surprising as such a high
complexity was to date only exhibited by a form of cops and robber
games called \emph{domination games} (see
\cite{FominKM03,FominGT11,KreutzerOrd09}). In these games, each cop
not only occupies his current vertex (as in other such games) but a whole neighbourhood of
fixed radius, which essentially allows to simulate set quantification
making the problem \pspace-complete.
The \dagw game, however, is to the best of our
knowledge the only cops and robber game with the usual
capturing condition that exhibits such a complexity.

With the same proof technique we also show that there are classes of
graphs for which any DAG decomposition of optimal width contains a
super-polynomial number of bags. (If \np${}\not={}$\pspace, this would
follow from the previous result, but we show this unconditionally.)
Furthermore, we obtain that for every $\epsilon\in(0,1)$ there are
graphs of DAG-width~$k$ with no polynomial size DAG decomposition
of width at most $k+k^{1-\epsilon}$.

\section{DAG-Width and Cops and Robber Games}
\label{sec:prelims}

\subsection{Preliminaries}
We assume familiarity with basic concepts of graph theory and refer
to~\cite{Diestel12} for background.  All graphs in this paper are
finite, directed and simple, i.e.~they do not have loops or multiple
edges between the same pair of vertices. Undirected graphs are
directed graphs with a symmetric edge relation.  If~$G$ is a graph,
then $V(G)$ is its set of vertices and $E(G)$ is its set of edges. For
a set $X\subseteq V(G)$ we write $G[X]$ for the subgraph of~$G$
induced by~$X$ and $G-X$ for $G[V(G)\setminus X]$. If $X$ is a set of
vertices, we write $\Reach_{G}(X)$ to denote the set of vertices
reachable from a vertex in~$X$. If $X=\{v\}$, we write
$\Reach_G(v)$. If $w\in Reach_G(v)$, we also write $v\le w$ and $v<w$
for the irreflexive variant. A \emph{root} of a graph is a vertex
without incoming edges.
A strongly connected component of a digraph~$G$ is a maximal
subgraph~$C$ of~$G$ which is strongly connected, \ie~between any pair
$u,v\in V(C)$ there are directed paths from~$u$ to~$v$ and from~$v$
to~$u$. All components considered in this paper will be strong and
hence we simply write \emph{component}.  

\subsection{DAG-Width and Cops and Robber Games}
A \emph{DAG decomposition} of~$G$ is a tuple $(D, B)$ where $D$ is a
DAG and $B 
= \{B_d \mid d\in V(D)\}$ is a set of bags, \ie subsets of $V(G)$, such that
\begin{enumerate}
\item $\bigcup_{d\in V(D)} B_d = V(G)$,

\item for all $a,b,c \in D$, if $a < b < c$, then $B_a \cap B_c \subseteq B_b$,

\item for every source $r \in V(D)$, $\Reach_{G}(B_{\ge r}) = B_{\ge r}$
  where $B_{\ge r} = \bigcup_{r\le d} B_d$,

\item for each $(a,b)\in E(D)$, $\Reach_{G - (B_a \cap B_b)}(B_{\ge b}
  \setminus B_a) = B_{\ge b} \setminus B_a$.
\end{enumerate}
The width of $(D, B)$ is $\max_{d\in V(D)} |B_d|$ and its size is $|V(D)|$. The 
\emph{DAG-width} $\DAGW(G)$ of~$G$ is the
minimal width of a DAG decomposition of~$G$.

\paragraph{Cops and Robber Games} \Dagw can be characterised by a cops and
robber game played on a graph~$G$ by a team of cops and a robber. Each cop
occupies a vertex of~$G$ or is outside of the graph. The robber occupies a
component of the graph that arises if we delete vertices occupied by the
cops. Hence, a game position can be described by a pair $(C, R)$, where $C$ is
the set of vertices occupied by cops and~$R$ is robber component. At the
beginning the robber chooses an arbitrary component $R_0$ of the graph and the
game starts at position $(\emptyset, R_0)$. The game is played in rounds. In each
round, from a position $(C, R)$ the cops first announce their next move, \ie~the
set $C'\subseteq V(G)$ of vertices that they will occupy next, and remove the
cops from the vertices $C\setminus C'$ that will not be occupied. Based on the
triple $(C, C', R)$ the robber chooses his new component $R'$ that must be
reachable from $R$ via cop free paths, \ie~a path reachable from $R$ that does
not contain a vertex occupied by a cop which remains on the board.  This
completes a round and the play continues at position $(C', R')$.

\paragraph{Abstract games}
An \emph{abstract game} is a tuple $(V,V_0,E,v_0,\Omega)$ where $(V,E)$ is a
directed graph, in which $V$ denotes the set of all positions and $E$
the set of moves, $V_0 \subseteq V$ is the set of positions in which
Player~0 has to move, $v_0\in V$ is the start position and $\Omega\subseteq
V^\omega$ is the winning condition. A \emph{play} is a maximal sequence $v_0, v_1, \ldots$ such that for all~$i\ge 0$, $(v_i,v_{i+1})\in
E$. Player~0 wins a play $\pi$ if it is finite and ends in a vertex
$v\in V_1\coloneqq V\setminus V_0$ without successors (so Player~1 has
to move, but cannot do this) or $\pi\in\Omega$. A (memoryless)
\emph{strategy} for Player~0 is a partial function $\sigma\colon V\to V$ such
that for all $v\in V$ where $\sigma$ is defined, $(v,\sigma(v))\in
E$. Strategies for Player~1 are defined analogously. A play $v_0,
v_1,\ldots$ is \emph{consistent} with $\sigma$ if for each $v_i\in V_0$ that
has a successor, we have $\sigma(v_i) = v_{i+1}$. We say that~$\sigma$
is \emph{winning} if Player 0 wins every play consistent with~$\sigma$
(and analogously for Player~1). We say that a game position
is \emph{consistent} with~$\sigma$ if there is a play consistent
with~$\sigma$ which contains the position.

\paragraph{The \dagw game} 
The \emph{\dagw game} $(V,V_0,E,v_0,\Omega)$ on a
graph~$G$ is defined as follows. The set of positions is $V \coloneqq \Pos(G)
\coloneqq \Pos_c \cupdot \Pos_r $ where 
\[\Pos_c \coloneqq V_0 \coloneqq  \{(C,R) \mid C\subseteq V(G),  R\subseteq V(G) 
\text{ is a component of }G- C\}\] 
are cop positions and 
\[\Pos_r \coloneqq \{(C,C',R) \mid C, C' \subseteq V(G) \text{ and }R\subseteq
V(G) \text{ is a component of }G-C\}\]
are robber positions. The set of robber moves is defned by 
\begin{align*}
\Moves_r(G) \coloneqq &\{\bigl( (C,C',R),(C',R') \bigr) \mid (C,C',R)\in \Pos_r, (C',R')\in
\Pos_c \\&\text{ and } R'\text{ is a component of } G-C' \text{ with
}R'\subseteq \Reach_{G-(C\cap C')}(R)\}
\end{align*}
and the set of cop moves by 
\[\Moves_c(G) \coloneqq \{\bigl( (C,R), (C,C',R) \bigr) \mid (C,R)\in \Pos_c,
(C,C',R)\in \Pos_r\}\cap \Mon\]
where \[\Mon \coloneqq \{\bigl((C,R),(C,C',R)\bigr)\mid
\Reach_{G-(C\cap C')}(R) = \Reach_{G-C}(R)\}\,,\]
is the \emph{monotonicity condition} saying that while making a move, it is forbidden 
to remove cops if this allows the robber to reach vertices that were unreachable for him 
otherwise. 
The whole set of moves is $\Moves(G) \coloneqq \Moves_c(G) \cupdot \Moves_r(G)$. 
The start position is $(\0,\0,\0)$ and the winning
condition for the cops (\ie for Player~0) is
$\Omega = \emptyset$, \ie the cops lose all infinite plays. 
As the cops always have a possible move (\eg by not changing their vertices), 
they win all finite plays. 

A cop is \emph{free} in a position $(C,R)$ if he is outside of the
graph (\ie $|C|<k$ in the game with $k$ cops) or on 
a vertex $v\in C$ such that $v\notin \Reach_{G-(C\setminus \{v\})}(R)$, \ie 
removing this cop does not lead to non-monotonicity and is thus allowed. We say
that the robber is in some set $A$ of vertices, if $A$ is contained in his component.


In~\cite{BerwangerDawHunKreObd12} Berwanger et al.\@ showed that~$k$
cops have a winning strategy~$\sigma$ on a graph if and only if the graph has a
DAG decomposition of width~$k$. From their proof one can easily infur
a correspondence between the number of bags in the decomposition and
the number of possible positions in the game where the cops play
according to $\sigma$.

\begin{lemma}\label{lemma:decs_and_positions}
If there is a DAG decomposition of a graph $G$ of width $k$ and of
size $n$, then~$k$ cops have a winning strategy such that the number
of positions consistent with this strategy is at most~$n\cdot |G|$.
\end{lemma}
\begin{proof}
  Consider the strategy obtained from a DAG
  decomposition of width $k$ as described in \cite[Theorem
  16]{BerwangerDawHunKreObd12}. In a play consistent with the strategy, the
  cops occupy only sets of vertices that correspond to some bag. Thus
  there are at most $n$ cop placements that can appear in a play. A
  position can be described by the cop placement and  an arbitrary vertex in the robber
  component. There are at most $|G|$ vertices, so the total number of
  positions is at most $n\cdot |G|$.
\end{proof}

\section{The Basic Construction}\label{sec:constr}

Let $s,t\colon \bbN \to \bbN$ be two non-decreasing
functions with $s(n) \in o(n)$ and let $n_0 \in \bbN\setminus\{0\}$
such that for all~$n\ge n_0$ we have $n-s(n)>2$. We also demand that
for all $n\ge n_0$, $s$ and $t$ satisfy $s(n)\ge 2$ and  $t(n)\ge 2$, and,
furthermore, if $t$ is bounded by a constant, then
$s(n)\in o(\frac{n}{\log n})$. Notice that $n_0$ depends only on $s$ and $t$, but not on~$n$.

We define a class of graphs $G_n(s,t)$ of \dagw~$n+1$ and size
$|V(G_n(s,t))| \in O(n^2\cdot t(n))$ (measured in the number of vertices)
such that every DAG decomposition of width $n+1$ has
super-polynomially many bags in the size of $G_n(s,t)$. The
parameters~$s$ and $t$ will be used to determine the difference
between the optimal width of a DAG decomposition and the best possible
width of a polynomial size decomposition. Our proof in the next section works already if
$s(n) = t(n) = 2$ for all~$n$ and the reader is invited to assume
these values at first. We shall consider what changes if~$s$ and~$t$
are different later.

For $n \in\{0,\ldots,n_0-1\}$, the graph $G_n(s,t)$ is a single
vertex without edges. For $n\ge n_0$, the graph $G_n(s,t)$ is
constructed as follows (see Figure~\ref{fig:graph}).  Let $M(n)$ and $C_i(n)$ for
$i\in\{0,\ldots,t(n)-1\}$ be pairwise disjoint sets of vertices, each of
$n-s(n)$ elements. Let $D^s(n)$ be a set of~$s(n)$ elements disjoint from
all $C_i(n)$ and $M(n)$, and let $N(n) = M(n) \cup D^s(n)$. Let $A^s(n)$ be a set
of~$s(n)$ new vertices, and let $B^t(n) = \{b_0(n),\ldots,b_{t(n)-1}(n)\}$ be a
set of~$t(n)$ new vertices.  The graph $G_n(s,t)$ has vertices
\[V(G_n(s,t)) = V(G_{n-s(n)-1}(s,t)) \cupdot A^s(n) \cupdot B^t(n) \cupdot
\bigcup_{i = 0}^{t(n)-1} C_i(n)\cupdot N(n)\,.\]
 We say that vertices from $A^s(n) \cup B^t(n) \cup \bigcup_{i = 0}^{t(n)-1} C_i(n) \cup N(n)$
are in level~$n\ge n_0$. If $n_1$ and $n_2$ are levels, then level~$n_1$ is
\emph{bigger} than level $n_2$ if $n_1>n_2$.

For a set $X$ let $X \choose 2$ be the set
$\{(x,y)\in X^2 \mid x\neq y\}$. The edges are defined by
\begin{align*}
E(G_n(s,t)) = &E(G_{n-s(n)-1}(s,t)) \cup {N(n) \choose 2} \cup \bigcup_{i=0}^{t(n)-1}
{C_i(n) \choose 2} \cup {A^s(n) \choose 2}\\
                     &\cup \bigcup_{i=0}^{t(n)-1} \Bigl( \bigl(N(n) \times
                     C_i(n)\bigr) \cup \bigl(C_i(n) \times D^s(n)
                     \bigr) \cup \bigl(C_i(n) \times
                     \{b_i(n)\}\bigr) \Bigr)\\
                    & \cup
                     \bigl(B^t(n) \times A^s(n)\bigr) \cup
                     \bigl(A^s(n) \times B^t(n) \bigr)\cup \bigl(A^s(n) \times M(n)\bigr)\\
                    & \cup \bigl(N(n) \times V(G_{n-s(n)-1}(s,t))\bigr) \cup \bigl(A^s(n)
                     \times V(G_{n-s(n)-1}(s,t))\bigr)\\
                    & \cup \bigl( V(G_{n-s(n)-1}(s,t))
                     \times A^s(n) \bigr) \cup \bigl(V(G_{n-s(n)-1}(s,t)) \times B^t(n)\bigr)\,.
\end{align*}
In other words, the first line says that $G_n(s,t)$ has all edges from
$G_{n-s(n)-1}(s,t)$, and that $N(n)$, all~$C_i(n)$
for $i\in\{0,\ldots,t(n)-1\}$ and $A^s(n)$ are cliques of sizes~$n$, $n-s(n)$
and~$s(n)$, respectively (\ie every two distinct vertices are connected in both directions). Note also that $B^t(n)$ induces an independent set.

\begin{figure}
  \centering
  \begin{tikzpicture}[scale=0.8]
    \node[draw,circle,inner sep=5mm] (C_0) at (0,0){$M$};
    \node[gray] at (0,-0.4){$n-s(n)$};

    \node[draw,circle,minimum height=5mm,minimum width=10mm] (D) at (2.5,0){$D^s$};
  \draw (C_0) -- (D);

    \node[draw,circle,inner sep=5mm] (C_1) at (3,-3){$C_0$};
    \node[gray] at (3,-3.4){$n-s(n)$};
      \node at (5,-3){$\cdots$};
    \node[draw,circle,inner sep=2.6mm] (C_t) at (7,-3){$C_{t(n)-1}$};
    \node[gray] at (7,-3.4){$n-s(n)$};
   \draw[-slim] (C_0) to (C_1);
   \draw[-slim] (C_0) to (C_t);
   \draw[-slim] (D) to (C_1);
   \draw[-slim] (D) to (C_t);
   \draw[-slim,bend right] (C_1) to (D);
   \draw[-slim,bend right] (C_t) to (D);

    \node[draw,circle,minimum height=5mm,minimum width=10mm] (A) at (-3,-2){$A^s$};
    \draw[-slim] (A) to (C_0);

  \node[vertex] (b_1) at (-4,-5){};
    \node at (-3,-5){$\cdots$};
  \node[vertex] (b_t) at (-2,-5){};
    \node[draw,ellipse,minimum height=5mm,minimum width=30mm] (B) at
    (-3,-5){};
  \node at (-5.2,-5.7){$B^t = \{b_0,\ldots,b_{t(n)-1}\}$};
  \draw (B) to (A);
  \draw[-slim,bend left=40] (C_1) to (b_1);
  \draw[-slim,bend left=40] (C_t) to (b_t);

   \node[draw,rectangle,minimum width=20mm,minimum height=15mm]
   (G_n-3) at (0,-3){$G_{n-s(n)-1}(s,t)$};
   \draw[-slim] (C_0) to (G_n-3);
   \draw[-slim] (D) to (G_n-3);
   \draw            (A) to (G_n-3);
   \draw[-slim] (G_n-3) to (B);

   \node[draw,ellipse,minimum height=2.4cm,minimum width=5cm,opacity=0.5] (N) at
    (1,0){};
    \node at (1.3,1){$N$};

  \end{tikzpicture}
  \caption{The construction of $G_n(s,t)$. Indices $\cdot(n)$ are
    omitted. A directed edge between two subgraphs such as $M$ and
    $C_0$ represents an edge from every vertex in $M$ to every vertex
    in $C_0$. An undirected edge represents edges in both directions.}
  \label{fig:graph}
\end{figure}

For the following lemma the precise definition of~$s$ and $t$ in
$G_{n-s(n)-1}(s,t)$ is inessential.

\begin{lemma}
  \label{lemma:dagw}
The \dagw of $G_n(s,t)$ is $n+1$ for $n\ge n_0$.
\end{lemma}
\begin{proof}
  The $n+1$ cops have the following winning strategy for the \dagw game on
  $G_n(s,t)$. The initial move of the robber must be to choose the whole graph
  as a component. Then the cops occupy $N(n)$ and we can assume that the
  robber chooses some $C_i(n)$ (for $i\in\{0,\ldots,t(n)-1\}$) because
  all other strongly connected components of $G_n(s,t) - N(n)$ have
  incoming edges from all~$C_i(n)$. (So the robber can go to every
  other component that is reachable now also later, see Lemma~5.21
  in~\cite{RabinovichPhD}.)  Then the remaining cop occupies
  $b_i(n)$. If the robber stays in $C_i(n)$, the cops from $M(n)$
  capture him there (recall that $M(n)$ has the same size $n-s(n)$ as
  every $C_i(n)$, $i\in\{0,\ldots,t(n)-1\}$). So we can assume that
  the robber goes to $A^s(n) \cup (B^t(n)\setminus\{b_i(n)\}) \cup
  V(G_{n-s(n)-1}(s,t))$. 
  The cops from $D^s(n)$ move to $A^s(n)$ and
  force the robber to proceed to $G_{n-s(n)-1}(s,t)$. If the robber
  remains in $B^t(n)$, he is captured in the next move. From now on,
  the $s(n)+1$ cops in $A^s(n) \cup \{b_i(n)\}$ stay there until the end
  of the play and the robber cannot leave $G_{n-s(n)-1}(s,t)$, which has
  outgoing edges only to $B^t(n)$ and to $A^s(n)$. The remaining $n-s(n)$ cops play in
  $G_{n-s(n)-1}(s,t)$ in the same way as on $G_n(s,t)$ until the
  robber is captured or expelled to a $b_j(n)$ for $j\neq i$. There he
  will be captured in one move.

  A winning robber strategy against~$n$ cops is to stay in $N(n)$
  until all~$n$ cops are there and then to go to $C_0(n)$. Due to
  the monotonicity of the winning condition, in that
  position of the game, no cop can be removed from his vertex as every
  vertex of $N(n)$ is reachable from every vertex of $C_0(n)$.
\end{proof}

\section{Big DAG Decompositions}\label{sec:bigdec}

In this section we prove that the described winning strategy for~$n+1$
cops from Section~\ref{sec:constr}
(let
us call it~$\sigma$) is the only possible one up to some irrelevant
changes. Then we count the number of positions that are consistent
with~$\sigma$ and observe that there are super-polynomially many of
them. It will follow that every DAG decomposition of the optimal width
has a super-polynomial size.

The first kind of change is to occupy the vertices within the sets
$M(n)$, $C_i(n)$ ($i\in\{0,\ldots,t(n)-1\}$), $A^s(n)$ and $D^s(n)$ or
to remove cops from them in
a different order than according to~$\sigma$. The second kind of a change is to
place cops on and then to remove them from vertices that are already
unavailable for the robber. (Note that~$\sigma$ never lets cops stay
on such vertices.) Both changes can obviously only increase the number of
possible positions.

\begin{observation}\label{obs:irrelevant-changes}
Let $\sigma'$ be as $\sigma$, but with some irrelevant changes applied.
Then $\sigma'$ uses as many cops as $\sigma$ and there are at least as
many positions consistent with $\sigma'$ as with~$\sigma$.
\end{observation}

\begin{lemma}\label{lemma:unique-sigma}
Up to irrelevant changes, there is only one winning strategy for $n+1$
cops on $G_n(s,t)$.
\end{lemma}
\begin{proof}
  We describe a family $\Gamma$ of robber strategies that enforces the
  cops to play according to~$\sigma$ (the cops strategy from
  Lemma~\ref{lemma:decs_and_positions}) up to irrelevant changes. If
  $n+1$ cops play in a different way, they lose. By
  Observation~\ref{obs:irrelevant-changes} we can ignore those
  possible changes. The strategies in $\Gamma$ differ only in the
  choice of components $C_i(\ell)$ (for all levels $\ell\ge n_0$ and
  $i\in\{1,\ldots,t(\ell)\}$) the robber visits during a play. Every
  choice is made independently of any other cop or robber move. Thus
  $\Gamma$ can be described as a set of strategies $\rho(I)$
  parameterised by a sequence of choices
  $I = i_n,i_{n-s(n)-1},\ldots,i_{n_0}$ (the indexes are the levels in
  decreasing order) where each $i_{\ell}$ is in
  $\{1,\ldots,t(\ell)\}$.

  Let some $I$ be fixed.  The robber remains in $N(n)$ until it is completely
  occupied by cops. If a cop was placed on a vertex $v\notin N(n)$ before $N(n)$
  was completely occupied, the cops lose. Indeed, consider the position where
  all vertices of $N(n)$ are occupied for the first time. Because~$v$ (whatever
  it is) has been occupied and because it is still reachable now from $N(n)$,
  the last $(n+1)$-st cop is still on~$v$, otherwise the monotonicity is
  violated at~$v$. The robber goes to some $C_i(n)$ from which~$v$ is reachable
  via paths avoiding $N(n)$ (such a $C_i(n)$ always exists) and the cops have no
  legal move. Thus the first moves of the cops are to occupy $N(n)$ and the last
  cop remains outside of the graph. (An irrelevant change can be made here:
  $N(n)$ can also be occupied in many steps. However, this ends in the same
  position in that the whole $N(n)$ is occupied. The second kind of irrelevant
  changes cannot ber applied here.)

  The robber chooses some $C_{i_n}(n)$ and the cops have no other possible move
  than to place the last remaining cop on $b_{i_n}(n)$ (otherwise we have the
  situation discussed in the previous paragraph). The robber goes to
  $A^s(n)\cup B^t(n)\setminus \{b_{i_n}(n)\}$. In this position, the cops in
  $\{b_{i_n}(n)\} \cup M(n)$ cannot be removed. So the cops from $D^s(n)$ must
  be used and they can be placed either in $A^s(n)$ or in $G_{n-s(n)-1}(s,t)$,
  or in $B^t(n)\setminus \{b_{i_n}(n)\}$, or in some $C_i(n)$. Placing the cops
  in $C_i(n)$ belongs to the second kind of irrelevant changes. In this case the
  robber component does not change, so after some number of such moves the cops
  have to play in a different way (or they lose). If at least one cop is placed
  in $G_{n-s(n)-1}(s,t)$ or in $B^t(n)\setminus \{b_{i_n}(n)\}$, the robber
  remains in $A^s(n)$ until all cops are placed. Then the cops have no legal
  move and lose. It follows that the cops from $D^s(n)$ must occupy the whole
  $A^s(n)$ (as above, regardless in which order) and the robber goes to
  $G_{n-s(n)-1}(s,t)$. From now on, all cops occupying $A^s(n)$ and $b_{i_n}(n)$
  will be reachable from the robber component and must stay there. It follows by
  induction on~$n$ that~$\sigma$ is the unique winning strategy for $n+1$ cops
  up to irrelevant changes.
\end{proof}

\begin{theorem}\label{thm:super-poly}
Every DAG decomposition of $G_n(s,t)$ of width $n+1$ has
super-po\-ly\-no\-mial\-ly many bags.
\end{theorem}
\begin{proof}
We count the number of positions that are consistent
with~$\sigma$. When the robber goes to the level~$n_0-1$, the cops are
occupying $A^s(\ell)$ for all levels $\ell\ge n_0$ that appear as indices of
$G_\ell( s,t)$. Additionally, for each $\ell$, the cops occupy exactly
one of $\{b_0(\ell),\ldots,b_{t(\ell)}(\ell)\}$. (If they occupy more
of them, the remaining cops do not suffice to capture the robber due
to Lemma~\ref{lemma:dagw}.)  Thus every $I$ induces a new position and 
there are $\prod_{\ell \text{ is a
    level}}t(\ell)$ possible positions with the robber in 
level~$n_0-1$, each corresponding to a particular choice of $C_i(m)$ for
$m>l$. By Lemma~\ref{lemma:decs_and_positions}, the number of bags in
  an DAG decomposition of optimal width is at least $\prod_{\ell \text{ is a
    level}}t(\ell)$. We set $s(\ell) = t(\ell)= 2$ for all
$\ell\in\bbN$ and $n_0=5$. Then
$\prod_{\ell \text{ is a level}}t(\ell) \ge 2^{\lfloor n/2\rfloor -
  5}$. On the other hand the size of $G_n(s,t)$ is 
\begin{align*}
|V(G_n(s,t))|& = |N(n)| + t(n)\cdot |C_i(n)| + |A^s(n)| +
|B^t(n)| + |G_{n-s(n)-1}(s,t)|  \\
&= n + t(n)\cdot (n-s(n)) + s(n) + t(n) + |G_{n-s(n)-1}(s,t)|\\
&= O(n^2\cdot t(n)) = O(n^2)\,,
\end{align*}
 so $ 2^{\lfloor n/2\rfloor - 2}$ is super-polynomial in
$|G_n(2,2)|$. 
\end{proof}

\section{Consider an Additive Constant Error}

In the simplest case we can set $s(\ell)= t(\ell) = 2$ for all
levels. Then we obtain at least $\lfloor n/2 \rfloor$ levels and the
size of an optimal decomposition is at least
$2^{O(\sqrt{|G_n(s,t)|})}$.  However, at the cost of one additional
cop we can construct a DAG decomposition with polynomially many
bags. We change~$\sigma$ to occupy $A^s(n)$ with two cops instead of
placing one cop on $b_i(n)$. We need one extra cop for this, but this
is not repeated in each level. Already in the first level, when the
robber goes to $G_{n-s(n)-1}(s,t) = G_{n-3}(s,t)$, we have cops only
on $A^s(n)$, but not
on~$b_i(n)$. So one cop is saved for $G_{n-3}(s,t)$ and we can
continue to play in all levels in the same manner.

We can change our choice of~$s$ and $t$ to make the number of
additional cops needed to obtain a polynomial size decomposition
unbounded. Let $s(\ell) = t(\ell) = \lfloor \ell/\log \ell\rfloor$ for
all~$\ell$. Then there are at least $\log n$ levels and
\begin{align*}
  |G_n(s,t)| = &\sum_{i=0}^{t(n)-1}(|C_i(n)| + |\{b_i(n)\}|) + N(n) +
                                                     |A^s(n)|  + |G_{n-s(n)}(s,t)|\\
                      & = (n -s(n) + 1)\cdot t(n) + n + s(n) +
                                                     |G_{n-s(n)-1}(s,t)|\\
                      & = O\left(\frac{n^2}{\log n}\right) + |G_{n-s(n)-1}(s,t)|
                                  = O(n^2)\,.
\end{align*}
It remains to estimate the number of bags in an optimal
decomposition. Let $n_1,n_2,\ldots$ be the indices in $G_{n_i}(s,t)$
appearing in $G_n(s,t)$, \ie $n_0 = n$, and for $i>0$ we have $n_i =
n_{i-1} - \lfloor n_{i-1}/\log n_{i-1} -1\rfloor$. Then for $n\ge 5$, \[n_i \ge n -
i\cdot n/\log n - i\,,\] which is easy to prove by induction
on~$i$. For all $i\le \log n/2$ we have 
\[n_i \ge n - \frac{\log n}{2}\cdot \frac{n}{\log n} -\frac{\log n}{2}
= \frac{n-\log n}{2} \] 
and thus for $n\ge 5$,
 \[t(n_i) =
\left\lfloor\frac{n_i}{\log n_i}\right\rfloor \ge \frac{(n-\log n)}{2\log n} \ge
\frac{n-n/2}{2\log n} \ge \frac{n}{4\log n}\,.\] So for $\lfloor\log
n/2\rfloor$ many levels $t(n_i) \ge n/(4\log n)$ and thus the number
of bags in any DAG decomposition is at least $(\frac{n}{4\log
  n})^{\lfloor\log n/2\rfloor}$.

We can define a winning cop strategy with only
polynomially many positions with the same trick as before investing
$s(n)-1$ new cops, \ie using $n+s(n)$ cops. Occupy $N(n)$ and when the
robber goes to some $C_i(n)$, occupy $A^s(n)$. The robber has to go to
the lower levels (otherwise he will be captured in
$C_i(n)\cup\{b_i(n)\}$ or in $B^t(n)$) and we do not need cops in
$B^t(n)$. In the following theorem we show that less than $n+s(n)$ cops
do not have a winning strategy with polynomially many positions. Thus
there is no polynomial approximation of an optimal DAG decomposition
by an additive constant.

\begin{theorem}
  For all $G_n(s,t)$ with $n\ge 25$, every DAG decomposition of width at most $n-s(n)-1$
  has size at least $\left(\frac{\log n}{16}\right)^{\log n/4}$.
\end{theorem}
\begin{proof}
  We describe a robber strategy against $n-s(n)-1$ cops that allows
  him to enforce at least $(\frac{n}{4\log n})^{\lfloor\log
    n/2\rfloor}$ positions (dependent on his choices of
  $C_i(\ell)$). The robber waits in $N(n)$ until it is occupied
  by~$n$ cops and goes to some $C_i(n)$ for $i\ge 0$ such that $b_i(n)$
  is not occupied by the cops. As $|C_i(n)| = n-s(n) \ge s(n)$ (recall that
  $n\ge 25$) and only $s(n)-1$ cops are left, in all
  $C_i(n)$ there is a cop free vertex. Similarly, the remaining~$s(n)-1$
  cops cannot occupy all $b_i(n)$ (there are $t(n) = s(n)$ many of them), so going to such a $C_i(n)$ is
  possible. Now the cops in $N(n)$ cannot move, $s(n)-1$ free cops
  cannot expel the robber from $C_i(n)$ and the robber waits in
  $C_i(n)$ for $b_i(n)$ to be occupied. When the cops announce to do
  this, he runs via $b_i(n)$ and $A^s(n)$ (which also has a free
  vertex) to $G_{n-s(n)-1}(s,t)$ and plays there in the same
  way recursively.

  If the cops do not occupy all vertices in $A^s(n)$ when the robber
  is in the subgraph $G_{n-s(n)-1}(s,t)$, they cannot use the cops
  from $M(n)$, so they cannot expel the robber from $M(n-s(n)-1)$
  (\ie $M$ in the highest but one level). Indeed, $|M(n-s(n)-1)| =
  (n-s(n)-1) - \left\lfloor \frac{n}{\log(n-s(n)-1)} \right\rfloor$
  and there are at most $(n+s(n)-1) - (n-s(n)) - 1 = 2s(n)-2$ free
  cops (namely $n-s(n)$ cops are in $M(n)$ and~$1$ cop is on
  $b_i(n)$), so we only need to choose an appropriately large~$n$, the
  least possible being~$25$. Hence we can assume that the
  cops occupy all $A^s(n)$, \ie $s(n)+1$ cops are tied in level~$n$ and
  there are at most $(n+s(n)-1) - s(n) - 1 = n-2$ cops for
  $G_{n-s(n)-1}(s,t)$.

  Let us count the number of possible positions that may appear in a
  play consistent with the described strategy of the robber. We first
  count the number of levels~$\ell$ where the cops have more than one
  cop in $B^t(\ell)$. When playing according to~$\sigma$, which uses
  $n+1$ cops, exactly one cop is in each $B^t(\ell)$ and now we have
  $s(n)+2$ cops more, so there are at least $\log n/2$ levels $i$ with
  $t(n_i)\ge n/(4\log n)$. Then there are $\log n/4$ levels $\ell$ with at
  most $4n\log^2 n$ cops in each of them. In order to cover
  $B^t(\ell)$ in each such level with $4n\log^2 n$ cops, we need
  $\frac{t(\ell)}{4n/\log^2 n}$ times. As $t(\ell)\ge n/(4\log n)$, we
  obtain $\frac{t(\ell)}{4n/\log^2 n} \ge \frac{\log n}{16}$. Summing
  up, there are $\log n / 4$ levels where the cops have to choose one
  among at least $\log n/16$ placements depending of the robber's
  choice of the corresponding $C_i(\ell)$. Thus the size of any DAG
  decomposition of width at most $n+s(n)-1$ is at least $(\log n/16
  )^{\log n/4}$. Recall that the size of $G_n(s,t)$ is polynomial in
  $n$ for $t(n)\le n$.

\end{proof}

\begin{corollary}
There is no polynomial size approximation of an optimal DAG
decomposition of $G_n(s,t) $ with an additive constant error. 
\end{corollary}

\section{\dagw is \pspace-complete}
\label{sec:pspace}

The construction of $G_n(s,t)$ shows how the robber can save the history of the
play in the current position. In each level $\ell$ he chooses one of the
$C_i(\ell)$, which is stored as a cop occupying $b_i(\ell)$ until the end of the
play. In the first step we extend the construction to reduce the \taut problem
to \DAGWprob (the problem, given a graph~$G$ and a number~$k$, is
$\DAGW(G)\le k$?).  \taut is the problem, given a formula of propositional
logic, to decide whether it is satisfied by all variable interpretations. In
general, \taut is \conp-hard~\cite[Problem {[}LO8{]}]{GareyJoh79}, but in our
version the formula is given in CNF. This is a restriction, as CNF-\taut is in
\ptime, but we use it only as a part of our construction which proves that
\DAGWprob is \pspace-hard. Later we extend it to reduce QBF (which is
\pspace-complete) to \DAGWprob where CNF is the general case~\cite[Theorem
7.10]{GareyJoh79} and is more convenient for our purposes. For a technical
reason we also restrict the formulae by forbidding a variable to appear twice in
a clause.

Let $\phi$ be a formula with~$n$ variables and $m$ clauses. The graphs $H_\phi$
are based on the graphs $G_{n_0+n-1}(s,t)$ from Section~\ref{sec:constr} such that
\begin{itemize}
\item (obviously) the number $n$ of non-trivial levels in $G_{n_0+n-1}(s,t)$ is
  the number of variables in~$\phi$,

\item $s(\ell) = t(\ell) = 2$ for all levels~$\ell$.
\end{itemize}

It will be convenient to define functions that relate variables and levels.
Let the variables of~$\phi$ be enumerated as $X_1$, $X_2$, \ldots, $X_n$. Let 
$\lev$ be a set of level numbers. Let $\var$ be the function mapping a level to an 
index of a variable such that $\var(\ell) = i$ if and only if $\lev(i) = \ell$.

Consider $G_{n_0+n-1}(s,t)$ as a starting point. We replace $G_{n_0-1}(s,t)$
(which is a single vertex) by the following gadget~$F_\phi$. It has a vertex $v$
and for every clause $C=L_1 \lor L_2 \lor{} \ldots {}\lor L_{r(C)}$ an $n$-clique
$K_C$ with vertices $v_1^C$, $v_2^C$, \ldots, $v^C_{r(C)}$. The edges go from
$v$ to every vertex of $K_C$ and back, \ie we have edges $(v,v_i^C)$ and
$(v_i^C,v)$ for all clauses $C$ and all $i\in\{1,\ldots,{r(C)}\}$. 

From outside of $F_\phi$, all vertices that had outgoing edges to the vertex of
$G_{n_0-1}(s,t)$ now have outgoing edges to all vertices of $F_\phi$. So those
edges build the set $A^s(\ell) \times F_\phi \cup N(\ell)\times F_\phi$ for
every level~$\ell\ge n_0$. Additionally we have edges from every vertex of
$F_\phi$ to every vertex of every $A^s(\ell)$, \ie the set of edges
$V(F_\phi)\times A^s(\ell)$ for every level~$\ell$. 

Finally, the edges from $K_C$ leaving $F_\phi$ reflect the
clause~$C$. For all levels~$\ell$ we add the following edges.
\begin{enumerate}
\item If $X_{\var(\ell)}$ does not appear in $C$, then there are edges
  $(v_j^C,b_0(\ell))$ and $(v_j^C,b_1(\ell))$ for all $v_j^C$ from $K_C$. 
\item If $X_{\var(\ell)} = L_i$ for some~$i\in\{1,\ldots,r(C)\}$ (\ie $X_{\var(\ell)}$ appears in $C$
  positively), then there is an edge $(v_{\var{\ell}}^C,b_1(\ell))$.
\item If $\neg X_{\var(\ell)} = L_i$ (\ie $X_{\var(\ell)}$ appears in $C$ negatively), then there is
  an edge $(v_{\var(\ell)}^C,b_0(\ell))$.
\end{enumerate}
\begin{observation}\label{obs:two-or-one}
By assumption, a variable can appear in a clause at most
once, so either it does not appear there at all, and we have two edges
from every vertex in the clique to the corresponding $B^t(\ell)$, or it appears once
and we have exactly one edge from the whole clique to $B^t(\ell)$.
\end{observation}

We claim that $n+1$ cops capture the robber in $H_\phi$ if and only if
$\phi$ is a tautology. 
\begin{lemma}\label{lemma:in-F-phi}
  \begin{enumerate}
  \item If $n+1$ cops win in $H_\phi$, then they must play according
    to the strategy~$\sigma$ as in the proof of Lemma~\ref{lemma:dagw} until
    a first position $P$ in that the robber component is $F_\phi$.
  \item The cops can play according to~$\sigma$ until position $P$
    (also if the robber wins).
  \item In position $P$ the only possible move for the cops is to occupy $v$
with the only free cop.
  \end{enumerate}
\end{lemma}
\begin{proof}
  The first two claims were proven in Lemma~\ref{lemma:dagw} and in
  Lemma~\ref{lemma:unique-sigma}. For the third claim, the only free
  cop is the one that would capture the robber on
$G_{n_0-1}(s,t)$ if we played on $G_{n_0+n-1}(s,t)$.
\end{proof}

The robber
chooses a component $C$. Let the new position be $P'$.
\begin{lemma}\label{lemma:use_one_cop}
  In position $P'$, $n+1$ cops win if and only if there is a
  cop on a vertex $w\in B^t(\ell)$ for some level $\ell$ such that
  \begin{enumerate}
  \item $w$ is not reachable from $C$ (via cop free
    paths) and 
  \item $X_{\var(\ell)}$ appears in $C$.
  \end{enumerate}
\end{lemma}
\begin{proof}
  Let $P'$ be some position as described above. If the cops
  win, in~$P'$ a cop is free  and
  the cops occupy all $A^s(\ell)$ (in all levels) and, in each level
  $\ell$, one of the two vertices $b_0(\ell)$ and $b_1(\ell)$. The
  free cop is placed on $v$ in $F_\phi$ (there is no other legal move
  that cops occupy a vertex in $F_\phi$ which does not lead to an immediate loss for the cops)
  and the robber chooses some $K_C$ for a clause
  $C=L_1 \lor L_2 \lor {}\ldots {}\lor L_r$. Let $X_{j_i}$ be the variable in
  the literal $L_i$. Now for all levels $\ell$ no cops from any
  $A^s(\ell)$ can be removed and reused as this would violate the
  monotonicity: there are edges from all vertices of $K_C$ to all
  vertices of all $A^s(\ell)$. The cops from every level $\ell$ such
  that $X_{\var (\ell)}$ does not appear in $C$ cannot be removed for
  the same reason. Thus only cops from $\bigcup_\ell B^t(\ell)$ can
  potentially be reused. If no cop from any $B^t(\ell)$ can be reused,
  the robber wins, otherwise the cops win as follows. Let this cop be
  in level $\ell$ and assume without loss of generality that
  $\var (\ell) = j_1$. The cop is placed on $v_2^C$, then the cop from
  $B^t(\lev(j_2))$ is placed on $v_3^C$. Now the cop from
  $B^t(\lev(j_3))$ is placed on $v_4^C$ and so on, the last cop
  occupying $v_1^C$.

  To see that $X_{\var(\ell)}$ must appear in $C$, note that otherwise the cop
  from $B^t(\ell)$ cannot be removed due to the monotonicity condition.
\end{proof}

\begin{lemma}\label{lemma:tautology}
  $n+1$ cops win on $H_\phi$ if and only if $\phi$ is a tautology.
\end{lemma} 
\begin{proof}
  Assume that $\phi$ is true under every valuation. The cops have the following
  winning strategy. Until the robber component is $F_\phi$ they play as
  described in Lemma~\ref{lemma:dagw}. Then by Lemma~\ref{lemma:in-F-phi} they
  have a free cop that is placed no $v\in V(F_\phi)$. The robber chooses some
  clique $K_C$ for a clause $C$.

  Let, without loss of generality, $R(C)\subseteq \{1,\ldots,r(C)\}$ be the set
  of indices~$\var(\ell)$ of $X_{\var(\ell)}$ appearing in $C$. Let
  $\alpha\colon \{X_i \mid i\in R\} \to \{0,1\}$ be the valuation of those
  $X_{\var(\ell)}$ defined by the choices of the robber during the previous part
  of the play as follows. For all levels~$\ell$, if $X_{\var(\ell)}$ appears in
  $C$ and the robber chose the component $C_i(\ell)$ in level $\ell$ (for
  $i\in\{0,1\}$), then $\alpha(X_{\var(\ell)}) = 1$ if $i=0$ and
  $\alpha(X_{\var(\ell)}) = 0$ if $i=1$. Let $\beta$ be a valuation of
  $X_1,\ldots,X_n$ extending $\alpha$. By assumption $\beta \models \phi$, so
  $\beta\models C$ and thus there is some $L_j$ in $C$ with $\beta \models L_j$.
  Let $\ell$ be such that $X_{\var(\ell)} = L_j$ or $\neg X_{\var(\ell)} = L_j$
  (we have $\var(\ell)\in R$). If $X_{\var(\ell)} = L_j$, then
  $\beta(X_{\var(\ell)}) = \alpha(X_{\var(\ell)}) = 1$ and thus the robber chose
  $C_0(\ell)$ in level~$\ell$. Then a cop occupies $b_0(\ell)$ in the current
  position. As $X_{\var(\ell)}$ occurs in $C$ positively, there is an edge from
  $K_C$ to $b_1(\ell)$ and by Observation~\ref{obs:two-or-one}, there is no edge
  from $K_C$ to $b_0(\ell)$. The cop from $b_0(\ell)$ can be reused and the cops
  win by Lemma~\ref{lemma:in-F-phi}. If $\neg X_{\var(\ell)} = L_j$, the situation is symmetric.

  If $\phi$ is not a tautology, let $\beta$ be a valuation with
  $\beta \not\models \phi$. The robber winning strategy is to choose in level
  $\ell$ the component $C_1(\ell)$ if $\beta(X_{\var(\ell)}) = 0$ and the
  component $C_0(\ell)$ otherwise. When the cops occupy $v$, the robber chooses
  the clique $K_C$ corresponding a clause $C = L_1 \lor{} \ldots {}\lor L_{r(C)}$ with
  $\beta \not \models C$. Then $\beta \not \models L_j$ for all
  $j\in\{1,\ldots,r\}$. Let $X_{i_j}$ be the variable in $L_j$. If
  $X_{i_j} = L_j$, then there is an edge from $K_C$ to $b_1(\lev(X_{i_j}))$. As
  $\beta\not\models X_{i_j}$, the robber chose $C_1(\lev(X_{i_j}))$, so there is
  a cop on $b_1(\lev(X_{i_j}))$, which cannot be removed. By
  Observation~\ref{obs:two-or-one} there is no other cop in
  $B^t(\lev(X_{i_j}))$. 

  By Lemma~\ref{lemma:use_one_cop} the cops from levels corresponding to
  variables that do not occur in $C$ cannot be reused. Thus all cops in all
  $B^t(\ell)$ is still reachable from the robber component and the cops lose.
\end{proof}

We extend our construction again to model choices of the cops that are still
recognisable at the end of the play. This leads to a reduction from \QBF, which
is \pspace-complete, to \DAGWprob. A \emph{quantified boolean formula} $\phi$ is
of the form \[\phi = Q_1X_1\ldots Q_rX_r \psi(X_1,\ldots,X_r)\] where $Q_i$ is
either $\forall$ or $\exists$ and $\psi$ is a propositional formula in CNF with
variables from $\calX =\{X_1,\ldots, X_r\}$.

The semantics of $\phi$ can be defined by means of a two-player game
with perfect information, which is convenient for our reduction. It is
the model-checking game $\MCGAME(\phi)$ for $\phi$ on the fixed structure
$(\{0,1\},\0)$ with no relations. The
players are called $\forall$ (the universal player) and $\exists$ (the
existential player). A play is played as follows. First, the
quantifier prefix of the formula is read from left to right and, for
$i=1,2,\ldots,r$, player $Q_i\in\{\forall, \exists\}$ chooses a value
$\beta(X_i) \in\{0,1\}$ for $X_i$. In other words, we have
positions~$P_j$ of the form
\[P_j = Q_jX_j,\ldots,Q_rX_r
\psi(X_1/\beta(X_1),\ldots,X_{j-1}/\beta(X_{j-1}),X_j,\ldots,X_r)\]
where $X_\ell/\beta(X_\ell)$ means that we replace all occurrences of
$X_\ell$ in $\psi$ by $\beta(X_\ell)$. If $Q_j = \forall$, then~$P_j$
is a position of the universal player, otherwise $P_j$ belongs to the
existential player. Successor positions have the form 
\[Q_{j+1}X_{j+1},\ldots,Q_rX_r
\psi(X_1/\beta(X_1),\ldots,X_j/\beta(X_j),X_{j+1},\ldots,X_r)\,.\]

The remaining positions of the game are of the form $(\theta,\beta)$
where $\theta$ is a subformula of~$\psi$ and $\beta$ is the valuation
of the variables as chosen in the first part of the play. The second part
starts in position $(\psi,\beta)$.  If $\theta = \theta_1 \lor \theta_2$,
the existential player moves to $(\theta_1,\beta)$ or to $(\theta_2,\beta)$
and if $\theta = \theta_1 \land \theta_2$, then the universal player moves
to $(\theta_1,\beta)$ or to $(\theta_2,\beta)$. In positions
$(X_i,\beta)$, the existential player wins if $\beta(X_i) = 1$ and
loses of $\beta(X_i) = 0$. In positions $(\neg X_i,\beta)$, the
universal player wins if $\beta(X_i) = 1$ and loses if $\beta(X_i) =
0$. The formula $\phi$ is true if and only of the existential player has a
winning strategy in the game.

It is very well known that deciding whether a given quantified formula
is true is \pspace-complete. 

The rest of the section is devoted to proof of the following theorem.

\begin{theorem}\label{thm:dagw-pspace-c}
 \DAGWprob is \pspace-complete.
\end{theorem}
The easier part is to show that \DAGWprob is in \pspace. It suffices to prove
that any play in the cops and robber game has polynomial length. Then deciding
the winner of the game is in \aptime (alternating \ptime) and thus in
\pspace. If~$k$ cops have a winning strategy on a graph, they also have a
winning strategy that always prescribes to place cops in a way that the space
available for the robber shrinks by at least one vertex. We consider a version
of the game where the cops \emph{have} to play in this manner. Then they win if
and only if they win in at most~$2n$ moves where~$n$ is the number of vertices
of the graph (the robber can also make $n$ moves). Thus any play lasts at most $2n$ steps.

For the hardness we reduce \QBF to \DAGWprob. Let \[\phi = Q_1
X_1\ Q_2 X_2\ldots Q_r X_r\ \psi(X_1,\ldots,X_r)\] be a quantified
boolean formula. Our construction of the graph $S_\phi$ extends the
construction of $H_\phi$. For every universal
quantifier we add a new level as in~$H_\phi$. For each existential quantifier we add a
level that is depicted in Figure~\ref{fig:existential}. The only
difference to a universal level is that we replace edges from
$N(\ell)$ to $C_i(\ell)$ by paths of length two which share the middle
vertex $c_i$. We now give a formal
description of the reduction.

\begin{figure}
  \centering
  \begin{tikzpicture}[scale=0.9]
    \node[draw,circle,inner sep=5mm] (C_0) at (0,0){$M$};
    \node[gray] at (0,-0.4){$n-2$};

    \node[draw,circle,minimum height=5mm,minimum width=10mm] (D) at 
(2.5,0){$D$};
  \draw (C_0) -- (D);

    \node[draw,circle,inner sep=5mm] (C_1) at (3,-3){$C_0$};
    \node[gray] at (3,-3.4){$n-2$};
    \node[draw,circle,inner sep=5mm] (C_t) at (6,-3){$C_1$};
    \node[gray] at (6,-3.4){$n-2$};
\node[vertex] (c_1) at (2,-1.5){};
\node at (1.8,-1.8){$c_0$};
\node[vertex] (c_t) at (5,-1.5){};
\node at (4.8,-1.8){$c_1$};

\draw[-slim] (C_0) to (c_1);
\draw[-slim] (c_1) to (C_1);
\draw[-slim] (D) to (c_1);

\draw[-slim] (C_0) to (c_t);
\draw[-slim] (c_t) to (C_t);
\draw[-slim] (D) to (c_t);

   \draw[-slim,bend right] (C_1) to (D);
   \draw[-slim,bend right] (C_t) to (D);

    \node[draw,circle,minimum height=5mm,minimum width=10mm] (A) at 
(-3,-2){$A$};
    \draw[-slim] (A) to (C_0);

  \node[vertex] (b_1) at (-4,-5){};
  \node[vertex] (b_t) at (-2,-5){};
    \node[draw,ellipse,minimum height=5mm,minimum width=30mm] (B) at
    (-3,-5){};
  \node at (-4.9,-5.5){$B = \{b_0,b_1\}$};
  \draw (B) to (A);
  \draw[-slim,bend left=40] (C_1) to (b_1);
  \draw[-slim,bend left=40] (C_t) to (b_t);

   \node[draw,rectangle,minimum width=20mm,minimum height=15mm]
   (G_n-3) at (0,-3){$S^{j+1}_\phi$};
   \draw[-slim] (C_0) to (G_n-3);
   \draw[-slim,bend right] (D) to (G_n-3);
   \draw            (A) to (G_n-3);
   \draw[-slim] (G_n-3) to (B);

  \end{tikzpicture}
  \caption{The existential level of $S_\phi$. Indices $\cdot(n)$ are omitted.}
  \label{fig:existential}
\end{figure}

If $\phi$ has no variables, then if $\phi$ is true, $S_\phi$ is a single 
vertex, and if $\phi$ is false, $S_\phi$ is a 2-clique. (So one cop wins if and 
only if $\phi$ is true.) Otherwise we start the construction of $S_\phi$ with 
$F_\psi$ and for $j=r,r-1,\ldots,1$ we construct graphs $S_\phi^j$ 
such that $S_\phi^1 = S_\phi$. Assume that $S_\phi^{j+1}$ is already 
constructed, then $S_\phi^j$ is the following graph. There are two cases. If 
$Q_j = \exists$, then the vertex set is 
\[V(S_\phi^j) = V_\exists(j) = V(S_\phi^{j+1}) \cupdot A(j) \cupdot B(j) \cupdot
C_0(j)\cupdot C_1(j) \cupdot N(n)\cupdot\{c_0(j),c_1(j)\}\,.\] 
Here $N(j) = M(j)\cupdot D(j)$ and $B(j)$ 
are as $N(j)$, $M(j)$, $D^s(j)$ and $B^t(j)$ in $G_n(s,t)$, \ie \[|B(j)| = 
|D(j)| = 2, |C_i(1)| = |M(1)| = 4, |C_i(k+1)| = |M(k+1)| = |M(k)| + 3\] for 
all $k\in\{2,\ldots,j\}$ and $i\in\{0,1\}$. Furthermore, $B(j) = 
\{b_0(j), b_1(j)\}$. Analogously to the graphs $G_n$ we call the set of vertices
$V_\exists\setminus V(S_\phi^{j+1})$  an \emph{existential level.}

The set of edges is
\begin{align*}
                    E(S_\phi^j) & = E(S_\phi^{j+1}) \cup {N(j) \choose 2} \cup 
\bigcup_{i=0}^{1}
{C_i(j) \choose 2} \cup {A(j) \choose 2}\\
                     &\cup \bigcup_{i=0}^1 \Bigl( \bigl(N(j) \times  
\{c_i(j)\}\bigr) \cup
                     \bigl(\{c_i(j)\} \times C_i(j)\bigr)\\
                     & \hspace{3cm}\cup \bigl(C_i(j) \times D(j)
                     \bigr) \cup \bigl(C_i(j) \times
                     \{b_i(j)\}\bigr) \Bigr) \\
                    & \cup
                     \bigl(B(j) \times A(j)\bigr) \cup
                     \bigl(A(j) \times B(j) \bigr)\cup \bigl(A(j) \times 
M(j)\bigr)\\
                    & \cup \bigl(N(j) \times V(S_\phi^{j+1})\bigr) \cup 
\bigl(A(j)
                     \times V(S_\phi^{j+1})\bigr)\\
                    & \cup \bigl( V(S_\phi^{j+1})
                     \times A(j) \bigr) \cup E(j)\,.
\end{align*}
Here $E(j)$, the edges connecting $F_\psi$ to the new level are
defined as follows. Let $K_C = \{v_1^C,\ldots,v_{r(C)}^C\}$ be a
clique in $F_\psi$ corresponding to a clause $C = L_1\lor{}\ldots
{}\lor L_r$. If $X_j = L_i$, then $(v_i^C,b_1(j))\in E(j)$. If $\neg X_j =
L_i$, then $(v_i^C,b_0(j))\in E(j)$. Otherwise (\ie if $X_j$ does not
appear in $C$) $\{(v_i^C,b_0(j)),(v_i^C,b_1(j))\}\subseteq E(j)$ for all $i\in\{1,\ldots,r(C)\}$.

In the second case $Q_j = \forall$. Then $V(S_\phi(j)) = V_\forall(j) = 
V_\exists(j) \setminus \{c_0(j),c_1(j)\}$ and the edges are defined as in 
$H_\phi$. We call the set of the new vertices a \emph{universal level.} 

We are going to show that $r+1$ cops win on $S_\phi$ if and only if the 
existential player wins $\MCGAME(\phi)$. For that we need some 
lemmata. Our global assumption is that there are $r+1$ cops in total.

Let $L = \{\ell_1,\ldots,\ell_r\}$ be the
set of level numbers and let $b = (b_{\ell_1},\ldots, b_{\ell_r})$ be a tuple of
bits. (Recall that the levels are numbered such that the biggest is in
$S^r_\phi$, but not in $S^{r-1}_\phi$.) Define
\[O^b(\ell) = \left( \bigcup_{\ell'>\ell} A(\ell')\right)\cup \{b_i(\ell_j) \mid
b_{\ell_j} = i, \ell_j>\ell\}\,.\]


The first lemma states what the cops can achieve in an existential level.

\begin{lemma}\label{lemma:cops_invariant}
  Let $\ell$ be the number of an existential level.  Let the cops occupy
  $O^b(\ell)$ for some bit vector $b$ and let the robber component be
  $S^{\ell}_\phi$.  For both $i\in\{0,1\}$ the cops have a strategy that allows
  them either to capture the robber or to expel him from level~$\ell$ such that
  the cops occupy precisely $O^b(\ell)\cup A(\ell) \cup \{b_i(\ell)\}$.
\end{lemma}
\begin{proof}
  Note that all cops in $O^b$ are all tied and there are $|N(\ell)|+1$ free
  cops.  The strategy is as follows. One cop is placed on $c_{1-i}(\ell)$. This
  creates two components of $S_\phi^\ell$: the one induced by $C_0(\ell)$ and
  the one induced by $N(\ell)$, $A(\ell)$, $B(\ell)$, $C_i(\ell)$ and
  $S_\phi^{\ell-3}$. If the robber is in $C_i(\ell)$, the remaining free cops
  expel him from there and the robber is in the other component. Then the cops
  occupy $N(\ell)$ and then the cop from $c_{1-i}(\ell)$ occupies
  $b_i(\ell)$. If the robber is in $C_i(\ell)$ and stays there, he is captured
  there by the cops from $N(\ell)$, so assume that the robber goes either to the
  component induced by $b_{1-i}(\ell)$ or to the component induced by $A(\ell)$
  and $S_\phi^{\ell-3}$. In any case the cops leave $D(\ell)$ and occupy
  $A(\ell)$. If the robber remains in $b_{1-i}(\ell)$, he is captured by the cop
  from $b_i(\ell)$, so assume that he goes to $S_\phi^{\ell-3}$. We obtain the
  required position.
\end{proof}

The following lemma describes what the robber can achieve in a level.

\begin{lemma}\label{lemma:robber_invariant}
  Let $\ell$ be the number of a level.  Let the cops occupy
  $O^b(\ell)$ for some bit vector $b$ and let the robber component be
  $S^{\ell}_\phi$. The robber has a strategy that permits him either to win or
  to reach a position where the robber is in $S_\phi^{\ell-3}$ and the cops
  occupy $O^b(\ell)$, $A(\ell)$ and at least one of $b_i(\ell)$. If the level is
  universal, the robber can additionally enforce $b_0(\ell)$ or $b_1(\ell)$. Furthermore,
  this robber strategy is winning if there are only $|N(\ell)|$ free cops.
\end{lemma}
\begin{proof}
  Again, all cops in $O^b$ are all tied and there are $|N(\ell)|+1$ free cops.
  The robber stays in $N(\ell)$ until it is completely occupied by the cops. In
  that position, one of $c_i(\ell)$ is not occupied by cops, and the robber runs
  to $C_i(\ell)$ and plays as in the proof of
  Theorem~\ref{thm:super-poly}. Note that the cops from $A(\ell')$ and from
  $B(\ell')$ for all $\ell'>\ell$ cannot be removed.
\end{proof}

\begin{lemma}
  There is a winning strategy for $r+1$ cops on $S_\phi$ if and only if
  $\phi$ is true.
\end{lemma}
\begin{proof}
  Assume that~$r$ cops have a winning strategy~$\sigma$ on~$S_\phi$. Without
  loss of generality we assume that one cop is placed in $B(\ell)$ in every
  level $\ell$, even if not enforced by the robber. In $\MCGAME(\phi)$, the
  existential player simulates the cops and robber game on $S_\phi$ by
  translating the moves of the universal player into robber moves and
  translating cop moves (according to~$\sigma$) into his choices of the
  existentially quantified variables. Assume that we reached a position~$P$ in
  the cops and robber game and a position $P_i$ (for $i\ge 1$) in the
  $\MCGAME(\phi)$ such that the following invariant (INV) holds.
\begin{itemize}
 \item Exactly the values $b_j = \beta(X_j)$ for the variables $X_j$ where 
$j\in\{1,\ldots,i-1\}$ are already chosen;

\item in the cops and robber game, the robber component is $S_\phi^{r+1-i}$ with
  the biggest level number $\ell = \ell(i)$;

\item the cops occupy $O^b(\ell)$ for the tuple $b =
  (b_1,\ldots,b_{i-1})$ and nothing else. 
\end{itemize}
Then there are exactly $\ell+1$ free cops. If $Q_i = \forall$, the
universal player chooses a value $b_i = \beta(X_i)$ for $X_i$. Then
the existential player simulates the cops and robber game playing for
the cops according to~$\sigma$ from position~$P$ and for the robber as in
Lemma~\ref{lemma:robber_invariant} such that the robber is expelled
from level $\ell$ and $A(\ell)$ and $b_i(\ell)$ are occupied by cops. The number of
free cops suffices for that. It is straightforeward to check that the
above invariant holds for $i+1$ and for the position of the cops and
robber game where the robber is blocked in the next level, \ie in level~$\ell-3$.

If $Q_i = \exists$, the existential player simulates the cops and robber 
game from~$P$ until the robber is expelled from level~$\ell$ according to 
Lemma~\ref{lemma:cops_invariant}. Here, the cops play according
to~$\sigma$ and the robber plays arbitrarily, but such that he is not
captured in level~$\ell$. For example, the robber goes directly to
$S_\phi^{\ell -3}$. Again, there are enough free cops for 
the simulation. Then exactly one of $b_0(\ell)$ and $b_1(\ell)\}$ is  occupied
by a cop. If it is $b_0(\ell)$, the existential player sets $\beta(X_i) = 0$, 
otherwise $\beta(X_i) = 1$. Again, the invariant holds.

When all variables have their values, the universal player chooses a
clause~$C$ and the existential player simulates in the cops and robber
game the move of the free cop to vertex~$v$ in $F_\psi$ and the move
of the robber to~$K_C$. 

Let us revise the current position. As the cops have played according to~$\sigma$
and~$\sigma$ is a winning strategy, there is still a free cop. The cops completely 
occupy every $A(\ell)$ and for all
$\ell$ exactly one of $b_0(\ell)$ and $b_1(\ell)$. Every vertex in all
$A(\ell)$ and $v$ are still reachable from the robber component, so the free cop
is in some $b_i(\ell)$. As the cop is free, there is no edge from any
$v_j^C$ to $b_i(\ell)$, \ie there is the edge from some $v_j^C$ to
$b_{1-i}(\ell)$. By construction of $S_\phi$, $X_j$ appears
in~$C$ and if $i=0$, then $X_j$ is negativ and if $i=1$, then
$X_j$ is positiv in~$C$. In the first case $b_i = 0$ and
$\beta(X_{v(\ell)}) = 0$, so~$C$ is satisfied and the existential
player wins by choosing $\neg X_j$. The second case is
symmetric.

For the other direction assume that the existential player has a winning 
strategy. We show that~$r+1$ cops have a winning strategy. The cops simulate 
the game $\MCGAME(\phi)$ while playing on $S_\phi$ by translating the moves of 
the 
robber to choices of the universal player and the choices of the existential 
player to their moves. Assume as before that (INV) holds for
some~$i\ge 1$.

There are two cases: level~$\ell$ is either existential or universal. In any
case, there are $\ell + 1$ free cops and if the level is universal,~$\ell$ of
them occupy $N(\ell)$. The robber escapes to some $C_i(\ell)$ for an
$i\in \{0,1\}$ or goes to the component consisting of $A(\ell)$, $B(\ell)$ and
$S_\phi^{\ell-3}$. If the robber is in $C_i(\ell)$ the last remaining cop is
placed on $b_i(\ell)$ and the robber proceeds to
$\{b_{1-i}(\ell)\}\cup A(\ell) \cup V(S_\phi^{\ell-3})$, otherwise the cop is
placed on $b_i(\ell) = b_0(\ell)$ and the robber is in
$\{b_1(\ell)\}\cup A(\ell) \cup V(S_\phi^{\ell-3})$. Now the cops from $D(\ell)$
occupy $A(\ell)$ and the robber goes to $S_\phi^{\ell-3}$ (if he goes to the
component induced by the free vertex of $B(\ell)$, he loses
immediately). Finally, the cops simulate the choice of the universal player:
$\beta(X_{v(\ell)}) = i$. It is easy to see that the invariant holds.

If the level is existential, the cops look up what value the strategy for the 
existential player in $\MCGAME(\phi)$ prescribes to choose for 
$X_{v(\ell)}$: $\beta(X_{v(\ell)}) = b_i \in \{0,1\}$. Then according to 
Lemma~\ref{lemma:cops_invariant} the cops can play such that the invariant 
holds again.

When the play arrives $F_\psi$, there is one free cop that is placed on~$v$ and 
the robber goes to some $K_C$ for a clause~$C$. As the existential player has a 
winning strategy, there is some literal $L_j$ in~$C$ that is satisfied 
by~$\beta$. Without loss of generality, $L_j = X_i$ and $\beta(X_i) = 1$. Then 
$b_1(\ell(i))$ is occupied by a cop, but there is no edge from~$v_i^C$ to 
$b_1(\ell(i))$. Recall that in $C$ every variable appears at most once, so 
there is no edge from~$C$ to $b_1(\ell(i))$ and the cop from $b_1(\ell(i))$ is 
free. The cops win by Lemma~\ref{lemma:use_one_cop}.
\end{proof}


\section{Conclusion}
We showed that \dagw cannot be computed efficiently in the classical
sense (assuming \ptime${}\neq{}$\pspace).  It would be interesting to
find (fixed-parameter tractable) algorithms computing constant factor
approximations of an optimal DAG decomposition. Another approach to
\dagw would be to show that \dagw and \kw are bounded in each other,
as deciding \kw is in \nptime. It is known that \dagw is
bounded in \kw by a quadratic function~\cite{KaiserKreRabSie14}.

\section*{References}
\bibliographystyle{elsarticle-harv}
\bibliography{bibl}

\end{document}